\documentclass[journal]{IEEEtran}
\usepackage[hidelinks]{hyperref}
\usepackage{amsthm}
\usepackage{amssymb}
\usepackage{mathtools}
\usepackage{mathrsfs}
\usepackage[justification=centering]{caption}
\usepackage{caption}
\usepackage{booktabs}
\usepackage{multirow}
\usepackage[table]{xcolor}
\usepackage[ruled,linesnumbered]{algorithm2e}
\usepackage{subfigure}
\usepackage{float}
\usepackage{bbm}
\usepackage{soul}
\usepackage{bm}
\usepackage{tabularx}
\usepackage{threeparttable}
\usepackage{amsmath}
\newtheorem{proposition}{Proposition} 
\usepackage{enumitem}

\usepackage{siunitx} 
\renewcommand{\arraystretch}{1.5}

\begin{document}
    
\title{
Design and Implementation of a High-Precision Wind-Estimation UAV with Onboard Sensors
}

\author{Haowen Yu$^{1}$, Na Fan$^{2}$, Xing Liu$^{3}$, Ximin Lyu$^{1, \dagger}$

\thanks{$^{1}$ School of Intelligent Systems Engineering, Sun Yat-sen University, Guangdong, China}
\thanks{$^{2}$ Department of Computer Science and Engineering, Hong Kong University of Science and Technology, Hong Kong, China}
\thanks{$^{3}$ Shenzhen ZEEY Technology Co., Ltd., Guangdong, China}

\thanks{$^\dagger$ denotes the corresponding authors: Ximin Lyu ({\tt\footnotesize lvxm6@mail.sysu.edu.cn}).}
  }

\maketitle

\begin{abstract}
Accurate real-time wind vector estimation is essential for enhancing the safety, navigation accuracy, and energy efficiency of unmanned aerial vehicles (UAVs). 
Traditional approaches rely on external sensors or simplify vehicle dynamics, which limits their applicability during agile flight or in resource-constrained platforms. 
This paper proposes a real-time wind estimation method based solely on onboard sensors.
The approach first estimates external aerodynamic forces using a disturbance observer (DOB), and then maps these forces to wind vectors using a thin-plate spline (TPS) model.
A custom-designed wind barrel mounted on the UAV enhances aerodynamic sensitivity, further improving estimation accuracy. 
The system is validated through comprehensive experiments in wind tunnels, indoor and outdoor flights. 
Experimental results demonstrate that the proposed method achieves consistently high-accuracy wind estimation across controlled and real-world conditions, with speed RMSEs as low as \SI{0.06}{m/s} in wind tunnel tests, \SI{0.22}{m/s} during outdoor hover, and below \SI{0.38}{m/s} in indoor and outdoor dynamic flights, and direction RMSEs under \ang{7.3} across all scenarios, outperforming existing baselines. Moreover, the method provides vertical wind estimates---unavailable in baselines---with RMSEs below \SI{0.17}{m/s} even during fast indoor translations.
\end{abstract}


\IEEEpeerreviewmaketitle

\section{INTRODUCTION}
Accurate real-time wind vector measurements, including wind speed and direction, are essential for unmanned aerial vehicle (UAV) operations.
These measurements offer three key benefits. 
First, they improve flight safety and stability by compensating for wind disturbances, especially in strong winds~\cite{2020CEP_Multiple_Guo}. 
Second, they enhance navigation precision by correcting wind-induced drift, thus supporting accurate trajectory tracking~\cite{2022SR_NeuralFly_OConnell}.
Third, they optimize energy use and flight endurance via wind-aware path planning and power management, maximizing operational range and flight time~\cite{2024IROS_EnergyPlanning_Duan}.
However, existing wind estimation methods often struggle to deliver high accuracy across a wide measurement range without compromising flight safety or onboard computational efficiency, especially in dynamic flight.

Traditional high-precision wind measurement methods rely on expensive or heavy external devices, e.g., ground anemometers, weather balloons, and weather radars.
Ground anemometers provide only point measurements~\cite{2020AAcous_Ultrasonic_Jia}.
Weather balloons are costly and subject to uncertainties due to spatial drift~\cite{2020AMS_WeatherBallon_Evert}. 
Weather radar lacks sufficient resolution for fine-scale wind measurements near operational aircraft~\cite{2014IMM_BriefOverview_Yeary}. 
Mounting wind sensors directly on UAVs can provide localized wind measurements, but typically increases power consumption and enlarges the UAV's geometry, which may negatively affect obstacle avoidance capability.
To address these challenges, researchers have explored methods that use only onboard sensors to estimate wind~\cite{2017JAOT_WindEstimation_Tilt_Palomaki,2022Atmosphere_Wind_Tilt_Meier}.
These methods typically infer wind from attitude changes and are based on quasi-static assumptions, neglecting full UAV dynamics. 
As a result, their accuracy degrades during dynamic flight maneuvers.

\begin{figure}[t]
\centering
\includegraphics[width=0.8\columnwidth]{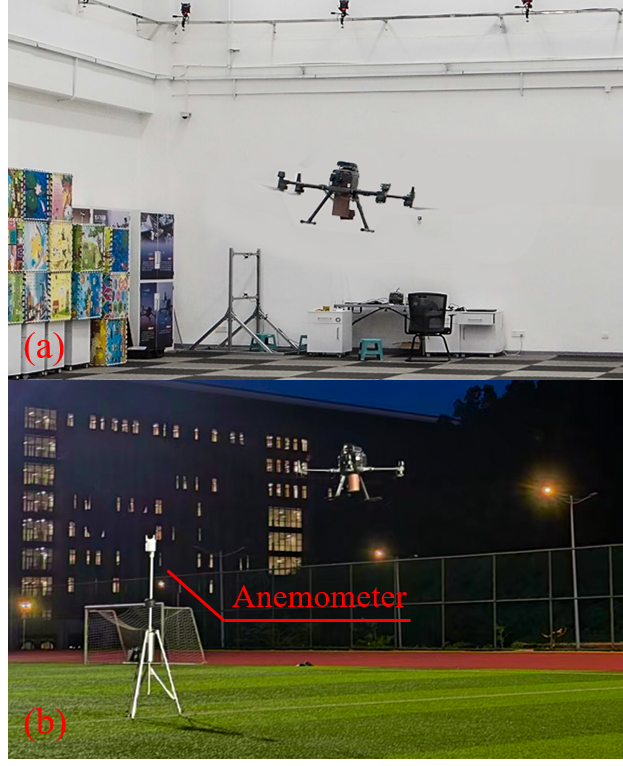}
\caption{ 
Wind estimation experiments in diverse environments:
(a) Indoor dynamic flight in still-air conditions;
(b) Outdoor hover and dynamic flight in natural, time-varying winds.
}
\label{fig:first_page_show}
\vspace{-0.49cm}
\end{figure}

To address these limitations, we propose a method that leverages only onboard sensors to estimate wind vectors, enabling accurate measurements during both static and dynamic flight.
This approach eliminates the need for additional dedicated measurement instruments, thereby reducing dead weight and complexity while enhancing safety and endurance. 

Our proposed method adopts a two-stage structure.
The front-end employs a disturbance observer (DOB)~\cite{2000TIE_NDOB_Chen,2018RAL_DOBHinfinity_Lyu} to estimate external forces acting on the vehicle. 
The DOB provides high-frequency force estimates without relying on quasi-static assumptions, enabling wind estimation during dynamic flight.
The back-end maps these estimated forces to wind vectors using a hybrid model pre-fitted on wind tunnel data: a thin-plate spline (TPS)~\cite{2003JRSS_TPS_Wood} for horizontal and a regression model for vertical components.
The TPS model’s minimum bending energy property ensures smooth and accurate force-to-wind mapping over a wide range.
To further enhance wind sensitivity and estimation accuracy, we incorporated a custom-designed wind barrel.

We evaluated the proposed wind estimation system across a variety of scenarios, including controlled wind tunnel tests, indoor flights, and outdoor hover and dynamic flight tests (see Fig.~\ref{fig:first_page_show}). 
Experimental results demonstrate that the proposed method consistently achieves high-accuracy wind estimation and outperforms baselines across all scenarios.
In wind tunnel tests, speed and direction RMSEs reached \SI{0.06}{m/s} and \ang{3.6} at a ground-truth wind speed of \SI{10}{m/s}.
During outdoor hover, RMSEs were \SI{0.22}{m/s} and \ang{3.3}, strongly correlated with the ground truth (correlation coefficient $r>0.9$).
In indoor dynamic flights, horizontal wind speed RMSE remained below \SI{0.38}{m/s} and direction RMSE below \ang{7.3}.
For the vertical component, where the ground-truth wind speed was \SI{0}{m/s}, the method achieved RMSEs under \SI{0.17}{m/s}.
Even in outdoor dynamic flights with varying natural wind, the method achieved RMSEs of \SI{0.29}{m/s} for speed and \ang{6.0} for direction in circular flights, consistently surpassing the baselines during unconstrained trajectories.

Our contributions are summarized as follows:
\begin{enumerate}
\item 
\textbf{DOB-based wind vector estimation algorithm:} 
We propose a real-time wind vector estimation method based on DOB, which relies solely on onboard sensors. 
The approach achieves high-precision wind speed and direction estimation without external sensors, and remains effective even during dynamic UAV maneuvers.

\item 
\textbf{Improved estimation accuracy and extended measurement range:} 
The DOB enables accurate external force estimation over a wide range.
The TPS model provides smooth and accurate mapping from force to wind.
The custom-designed wind barrel increases aerodynamic sensitivity, further enhancing estimation performance.

\item 
\textbf{Comprehensive experimental validation:} 
Extensive experiments, including wind tunnel tests, outdoor hover trials, and both indoor and outdoor dynamic flights, validate the method’s reliability and performance in static and dynamic scenarios (see Sec.~\ref{sec:Experiment}).
\end{enumerate}

The remainder of this paper is organized as follows: 
Sec.~\ref{sec:related_work} provides an overview of related work.
Sec.~\ref{sec:Preliminaries} introduces the defined coordinate frames and the UAV dynamics. 
Sec.~\ref{sec:Wind_Estimation_Methodology} presents the principle and framework of the proposed wind estimation system, including the DOB formulation (Sec.~\ref{sec:dob_design}), the wind barrel design and validation (Sec.~\ref{sec:wind_barrel_design}), and the force-to-wind mapping model (Sec.~\ref{sec:Force_wind_relationship}).
Experimental results are presented in Sec.~\ref{sec:Experiment}. 
Finally, Sec.~\ref{sec:conclusion} concludes the paper and introduces future work.

\section{RELATED WORK}  
\label{sec:related_work}
Pioneering UAV wind measurements used \textbf{fixed-wing aircraft} with Pitot tubes to infer airspeed
~\cite{2005ICRA_EarlyFixedWing_Kumon, 2020TAES_MovingHorizon_Wenz}.
However, Pitot tubes require prior knowledge of airflow direction and precise alignment~\cite{2015ICUAS_OnEstimation_Johansen}, and their accuracy degrades at low airspeeds~\cite{2015JP_PitotTubeCalibration_Buscarini}.
Furthermore, fixed-wing UAVs cannot hover, which restricts in-situ wind sensing at a fixed location. 
\textbf{Multirotors} have attracted increasing interest for wind sensing due to their hovering and maneuvering capabilities.
Some studies have integrated conventional wind sensors, such as ultrasonic anemometers~\cite{2018ICUAS_UltrasonicAnemometers_Hollenbeck} and multi-hole pressure probes~\cite{2018JWEIA_MHPP_Prudden}.
However, these sensors are often bulky and heavy, reducing UAV maneuverability, stability, and flight duration~\cite{2023BLM_WindEstUsingAnemometer_Li}.
For instance, mounting a commercial ultrasonic anemometer such as the Sniffer4D Mini 2~\cite{soarability2025sniffer4d} on a DJI M300 adds 780\,g of mass and requires 12\,W of continuous power. To avoid rotor wash, it must be mounted on a mast of approximately 0.37\,m above the rotor plane, significantly increasing the vehicle’s vertical extent.
Others have developed specialized sensors (whiskers~\cite{2020IROS_Whisker_Kim,2021ICRA_Whisker_Andrea}, microphones~\cite{2023Sensors_Microphones_Makaveev}, hot-wire sensors~\cite{2023ICRA_FlowDrone_HotWireFlowSensors_Simon}, surface wind sensors~\cite{2024Measurement_SurfaceSensor_Wang}).

\textbf{Model-based approaches} estimate wind indirectly from UAV motion states (e.g., attitude, velocity) and aerodynamic interactions~\cite{2017JAOT_WindEstimation_Tilt_Palomaki, 2015SAAP_RealTimeWindEstimation_TiltEarliest_Neumann, 2017AIAAAFM_Measuring_GonzalezRocha}.
However, they are typically limited to steady or near-hover flight and cannot estimate vertical wind components. 
Extensions for 3D wind estimation in constant-velocity flight exist~\cite{2022Atmosphere_Wind_Tilt_Meier,2019AIAASciTech_Model_UpDynamics_GonzalezRocha}, but they still struggle during rapid maneuvers and with uncertain aerodynamics.
Kalman Filter (KF)-based estimators are widely used to improve robustness by explicitly modeling process and measurement uncertainties~\cite{2019CCTA_KalmanFilter_Xing, 2020MSR_TwoStageKalman_Hajiyev, 2020IROS_TouchTheWind_ukf_NNlstm_Tagliabue}.
However, their accuracy strongly depends on the fidelity of the UAV dynamics and wind interaction models, and consequently, their performance is often degraded by real-world model imperfections~\cite{2020IJC_Using_Perozzi}.
\textbf{Learning-based methods}~\cite{2023ICRA_FlowDrone_HotWireFlowSensors_Simon, 2022Measurement_WindML1_Zimmerman, 2019AIAASciTech_Estimating_Allison, 2020JBSMSE_Hybrid_Marton} offer data-driven alternatives by mapping UAV sensor data (e.g., IMU, GPS) directly to wind vectors.
They leverage the nonlinear approximation capabilities of neural networks to capture complex relationships, 
but often require large-scale training data and exhibit poor generalization across sensors and platforms.

\textbf{Disturbance Observer (DOB)-based methods} have recently emerged as promising alternatives. 
They estimate external disturbance forces on the UAV and map them to wind vectors using aerodynamic models calibrated from flight data.
Lyu~\textit{et al.}~\cite{2018RAL_DOBHinfinity_Lyu} demonstrated that DOB-based control improves wind disturbance rejection in VTOL UAVs.
Building on this, Yu~\textit{et al.}~\cite{2024IROS_DOB_Yu} proposed a DOB-based wind estimation method using onboard sensors (IMU, GPS) to achieve real-time 3D wind vector estimation.

Despite these advances, existing approaches remain limited in terms of accuracy, capability for 3D wind estimation, robustness under dynamic flight conditions, data efficiency, and adaptability across sensor suites and platforms.
Our approach addresses these gaps by enabling accurate and reliable 3D wind estimation across diverse flight conditions, even during 
Beyond estimation, robust wind-aware control is critical for safe UAV operation. Recent advances in safety-guaranteed and fault-tolerant control—e.g., Yu et al.~\cite{yu2024fault} on fault-tolerant cooperative control and Yu et al.~\cite{yu2025review} on performance-guaranteed safety control—highlight the growing emphasis on resilience under environmental disturbances (including wind). While our work focuses on estimation, its high-fidelity wind output can serve as a key enabler for such control frameworks, closing the loop from perception to robust action.

\section{PRELIMINARIES}
\label{sec:Preliminaries}

\begin{figure}[t]
\vspace{0.0cm}
\centering
\includegraphics[width=1.0\columnwidth]{ 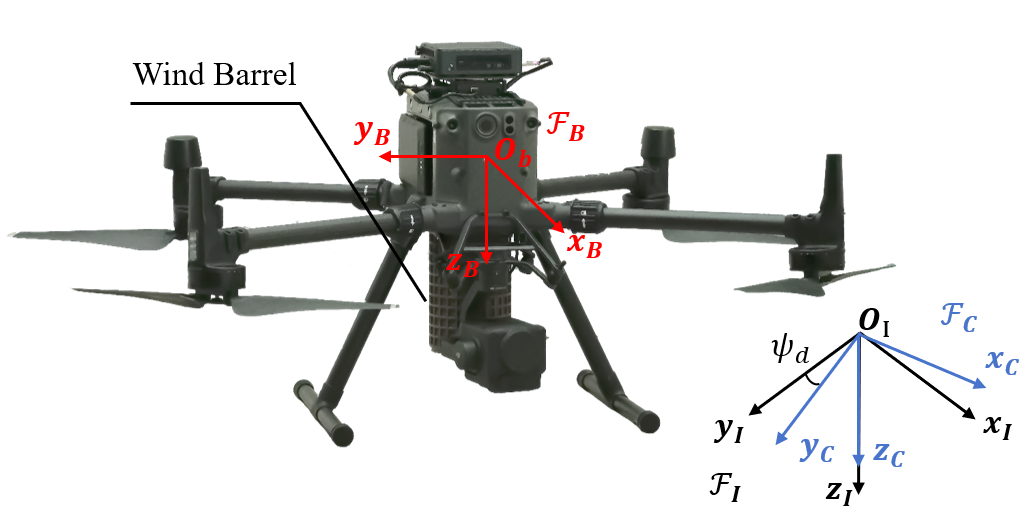}
\caption{
UAV platform with a custom wind barrel mounted beneath to enhance wind estimation.
Three coordinate frames are defined: inertial ($\mathcal{F}_I$), body ($\mathcal{F}_B$), and intermediate ($\mathcal{F}_C$).
}
\label{fig:m300_combine_coord}
\vspace{-0.0cm}
\end{figure}

\subsection{Coordinate Systems}
\label{subsec:coord}
To describe the UAV's motion and wind estimation process, three coordinate frames are introduced (see Fig.~\ref{fig:m300_combine_coord}):  
the Earth-fixed inertial frame $\mathcal{F}_I$ with axes $(\bm{x}_I, \bm{y}_I, \bm{z}_I)$;  
the body frame $\mathcal{F}_B$ with origin $\bm{o}_b$ at the UAV's center of gravity and axes $(\bm{x}_B, \bm{y}_B, \bm{z}_B)$;
and the intermediate frame $\mathcal{F}_C$ with axes $(\bm{x}_C, \bm{y}_C, \bm{z}_C)$, obtained by rotating $\mathcal{F}_I$ about its $\bm{z}_I$-axis by the desired yaw angle $\psi_d$.

\subsection{UAV Dynamics}
\label{subsec:dynamics}
The UAV's dynamics are described by~\cite{2018ARC_Review_Emran}:
\begin{equation}
\label{eq:dynamic1}
\begin{aligned}
    m\bm{\ddot{p}}&=-u_f\bm{R}\bm{e}_3+m{g}_0\bm{e}_3+\bm{f}_e,\\
    \bm{J}\bm{\ddot{\eta}}&=-\bm{\mathbb{C}(\eta,{\dot{\eta}}){\dot{\eta}}}+ {\bm{u}_{\tau}}+\bm{{\tau}}_e.
\end{aligned}
\end{equation}
where 
$m$ is the UAV's mass; 
$g_0$ is the gravitational acceleration;
$\bm{p} = [p_x, p_y, p_z]^\top$ is the position of $\bm{o}_b$ represented in $\mathcal{F}_I$;
$\bm{e}_3 = [0, 0, 1]^\top$ is the unit vector along $\bm{z}_I$;
$\bm{R} \in SO(3)$ is the rotation matrix that transforms vectors from $\mathcal{F}_B$ to $\mathcal{F}_I$; 
$\bm{f}_e = [f_{ex}, f_{ey}, f_{ez}]^\top$ is the total external force in $\mathcal{F}_I$;
$\bm{\eta} = [\phi, \theta, \psi]^\top$ represents the UAV's attitude (e.g., roll, pitch, and yaw angles);
$\bm{J} \in \mathbb{R}^{3 \times 3}$ is the UAV's moment of inertia in $\mathcal{F}_B$;
$\bm{\tau}_e \in \mathbb{R}^3$ is the total external torque in $\mathcal{F}_B$; 
$\mathbb{C}(\bm{\eta},\dot{\bm{\eta}})$ is the Coriolis matrix;  
${{u}_{f}} \in \mathbb{R}$ and $\bm{{{u}_{\bm{\tau}}}}\in \mathbb{R}^3$ are the total thrust and torque control inputs generated by the four rotors~\cite{2017TRO_External_motormodel3_Tomic}.

Fig.~\ref{fig:speed_thrust_curve} illustrates the relationship between the thrust and normalized RPM of each rotor.
The relationship between ${{u}_{f}}$ and normalized RPM can be approximated using a quadratic polynomial~\cite{2014AIAAAAC_ReynoldsNumberEffects_Deters}. 
For our UAV platform, the fitted mapping is:
\begin{equation}
\label{eq:thrust_curve}
{u_f} = \sum_{i=1}^4 (207\Omega_i^2+11.34\Omega_i+0.01315),
\end{equation}
where $\Omega_i$ is the normalized RPM of the $i$-th rotor (actual RPM divided by the maximum RPM). 
Note that while the full control input of the quadrotor includes both thrust $u_f$ and $\bm{u}_\tau$, only $u_f$ is required for the subsequent wind estimation process (see Sec.~\ref{sec:dob_design}). Therefore, we provide only the thrust–RPM relationship here.
The coefficients in Eq.~\ref{eq:thrust_curve} were obtained via static thrust-stand calibration: each motor was mounted on a load cell, and thrust was measured across 10 evenly spaced RPM levels from 0 to 3,650 RPM. The normalized speed is computed as $\Omega_i = \text{RPM}_i \times 1.047 \times 10^{-4}$, yielding a calibration range of $\Omega \in [0, 0.38]$.

\section{METHODOLOGY}
\label{sec:Wind_Estimation_Methodology}

Sec.~\ref{sec:wind_estimation_system_overview} presents a system overview, outlining data flow and module functions. 
Sec.~\ref{sec:stage1} details Stage I, which estimates external force via DOB (Sec.~\ref{sec:dob_design}) and improves aerodynamic sensitivity using a custom wind barrel (Sec.~\ref{sec:wind_barrel_design}).
Sec.~\ref{sec:Force_wind_relationship} describes Stage II, covering training data acquisition (Sec.~\ref{sec:wind_tunnel_data_collection}), modeling using TPS and polynomial regression (Sec.~\ref{subsec:data_fitting}), and real-time filtering (Sec.~\ref{sec:dynamice_filtering}).

\begin{figure}[t]
	\centering
	\includegraphics[width=1.0\columnwidth]{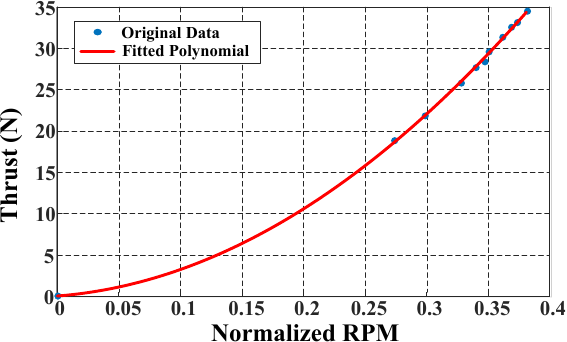}
 \caption{'Revolution-Thrust' curve for each motor.
		\label{fig:speed_thrust_curve}}
\end{figure}

\subsection{Wind Estimation System Overview}
\label{sec:wind_estimation_system_overview}

Our wind estimation system adopts a two-stage structure, as illustrated in Fig.~\ref{fig:system_design_frame}:

\textbf{Stage I: External Force Estimation (Fig.~\ref{fig:system_design_frame}(a))}: 
External aerodynamic forces $\bm{f}_e$ induced by ambient airflow are estimated by a disturbance observer (DOB). 
The DOB computes $\hat{\bm{f}}_e^I$ in the inertial frame $\mathcal{F}_I$ using measured thrust $u_f$, acceleration $\bm{\ddot{p}}$, and attitude $\bm{\eta}$ as inputs.
$\hat{\bm{f}}_e^I$ are then transformed into the intermediate frame $\mathcal{F}_C$ to obtain $\hat{\bm{f}}_e^C$ for subsequent wind estimation.

\textbf{Stage II: Force-to-Wind Mapping (Fig.~\ref{fig:system_design_frame}(b))}: 
The horizontal force components $(\hat{f}_{ex}^C, \hat{f}_{ey}^C)$ are used to compute the wind direction $\vartheta_r$, which is the same as the direction of the force. 
The horizontal wind speed $\hat{V}_{h}^C$ is then obtained from these components using a TPS model, while the vertical force $\hat{f}_{ez}^C$ is converted to vertical wind speed $\hat{V}_{v}^C$ via polynomial regression.
These components form the relative air velocity $\hat{\bm{A}}_r^C$, which is then transformed into the inertial frame $\mathcal{F}_I$ as $\hat{\bm{A}}_r^I$. 
Finally, the true wind vector $\hat{\bm{A}}_w^I$ is obtained by subtracting the UAV’s velocity $\dot{\bm{p}}$ from $\hat{\bm{A}}_r^I$.

\begin{figure*}[t]
\centering
\includegraphics[width=1.9\columnwidth]{ 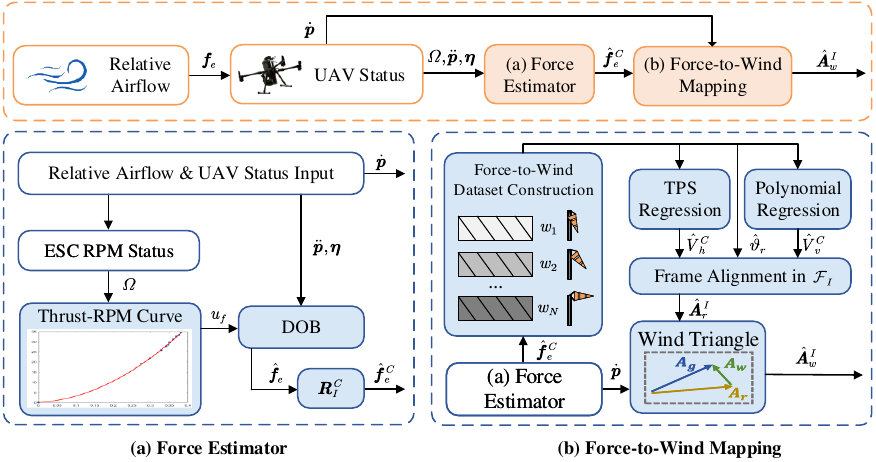}
\caption{
System Overview: 
(a) DOB-based external force estimation.  
(b) Force-to-wind mapping and wind inference.
}
\label{fig:system_design_frame}
\end{figure*}

\subsection{Stage I: DOB-Based External Force Estimation}
\label{sec:stage1}

\subsubsection{DOB Formulation}
\label{sec:dob_design}
We use a DOB to estimate the external disturbances caused by airflow in real time.
Our approach extends the work of Y{\"u}ksel~\textit{et al.}~\cite{2014AIM_NonlinearFourWrench_Yuksel}. 
The UAV’s configuration is described by $\bm{\Theta} = [\bm{p}^\top,\bm{\eta}^\top]^\top$, then we write \eqref{eq:dynamic1} in a compact form:
\begin{equation}
\label{eq:system_dynamic}
    \begin{aligned}
        \bm{d_e} = \bm{\mathcal{M}}\bm{\ddot{\Theta}}+\bm{\mathcal{C}}\bm{\dot{\Theta}}+\bm{\mathcal{D}}\bm{{\mu}}+\bm{g},
    \end{aligned}
\end{equation}
where 
\begin{equation}
\label{eq:dynamic_matrix_explain}
\begin{aligned}
    \bm{\mathcal{M}} &= \begin{bmatrix}
        m\mathbf{I}_3 & \mathbf{0}_{3 \times 3} \\
        \mathbf{0}_{3 \times 3} & \mathbf{J}
    \end{bmatrix}, \quad
    \bm{\mathcal{C}} = \begin{bmatrix}
        \mathbf{0}_{3 \times 3} & \mathbf{0}_{3 \times 3} \\
        \mathbf{0}_{3 \times 3} & \bm{\mathbb{C}}
    \end{bmatrix}, \\
    \bm{\mathcal{D}} &= \begin{bmatrix}
        \mathbf{R}\bm{e}_3 & \mathbf{0}_{3 \times 3} \\
        \mathbf{0}_{3 \times 3} & -\mathbf{I}_3
    \end{bmatrix}, \quad
    \bm{g} = \begin{bmatrix}
        -mg_0\bm{e}_3 \\
        \mathbf{0}_{3 \times 1} 
    \end{bmatrix}. 
\end{aligned}
\end{equation}
Here, 
$\bm{d_e} = [\bm{{f_e}}^\top, \bm{{\bm{\tau}_e}}^\top]^\top$ is the external wrench applied to the UAV;  
$\mathbf{I}_3\in {{\mathbb{R}}^{3\times 3}}$ is the identity matrix; 
$\bm{\mu} = [{u_f}, {\bm{u_\bm{\tau}}}^\top]^\top \in {{\mathbb{R}}^{4\times 1}}$ is the control input of total thrust and torque.

Following the formulation in~\cite{2000TIE_NDOB_Chen}, the DOB is described by:
\begin{equation}
\label{eq:we_derivative_from_chenwenhua}
\begin{aligned}
    \dot{\bm{\hat{d}}_e} = \bm{L}(\bm{\Theta}, \dot{\bm{\Theta}})\left(\bm{d}_e - \bm{\hat{d}}_e\right),
\end{aligned}
\end{equation}  
where 
$\bm{d}_e$ is the actual wrench as defined in \eqref{eq:system_dynamic};
$\bm{\hat{d}}_e$ is the estimated external wrench; 
and $\bm{L}(\bm{\Theta}, \dot{\bm{\Theta}}) \in {{\mathbb{R}}^{6\times 6}}$ is a gain matrix.
Since we do not assume any specific model for the external wrench and have no prior information about its time derivative~\cite{2000TIE_NDOB_Chen, 2014AIM_NonlinearFourWrench_Yuksel}, we assume:
\begin{equation}
\begin{aligned}
\label{eq:we_equals_zero0}
    \dot{\bm{d}}_e = \bm{0}.
\end{aligned}
\end{equation}
Define the observation error $\bm{e}$ as the difference between the actual external wrench $\bm{d_e}$ and its estimate $\bm{\hat{d}}_e$:
\begin{equation}
\begin{aligned}
\label{eq:observer_error}
    \bm{e} = \bm{d}_e - \bm{\hat{d}}_e.
\end{aligned}
\end{equation}  
Taking the derivative of \eqref{eq:observer_error} and substituting from \eqref{eq:we_derivative_from_chenwenhua}, we obtain:
\begin{equation}
\begin{aligned}
\label{eq:doteplusLe_equals_dotwe}
   \dot{\bm{e}} = -\bm{L}(\bm{\Theta}, \dot{\bm{\Theta}}) \bm{e}.
\end{aligned}
\end{equation}
The convergence of the observer error depends on the gain matrix $\bm{L}(\bm{\Theta}, \dot{\bm{\Theta}})$. 

To enable axis-wise tuning of convergence rates, we follow the design proposed in~\cite{2014AIM_NonlinearFourWrench_Yuksel} and define the gain matrix as: 
\begin{equation}
\begin{aligned}
\label{eq:Lqq_define}
    \bm{L}(\bm{\Theta}, \dot{\bm{\Theta}}) = \bm{K}_I \bm{\mathcal{M}}^{-1},
\end{aligned}
\end{equation}
where 
$\bm{K}_I \in \mathbb{R}^{6 \times 6}$ is a diagonal gain matrix with positive elements, allowing independent tuning of observer gains for each axis.

\begin{proposition}
Consider the wrench estimator defined by \eqref{eq:we_derivative_from_chenwenhua}. 
If the gain matrix is structured as \eqref{eq:Lqq_define},  then the estimated external wrench $\bm{\hat{d}}_e$ converges asymptotically to the actual external wrench $\bm{d_e}$, i.e., $\bm{\hat{d}}_e \rightarrow \bm{d_e}$.
\end{proposition}

\begin{proof}
We consider the Lyapunov function candidate:
\begin{equation}
\begin{aligned}
    \bm{V}(\bm{e}, \bm{\Theta}) = \bm{e}^\top \bm{\mathcal{M}} \bm{e},
\end{aligned}
\end{equation}
is positive definite.
Considering \eqref{eq:doteplusLe_equals_dotwe} and \eqref{eq:Lqq_define}, we can write:
\begin{equation}
\label{eq:Liyapunuofu}
\begin{aligned}
    \frac{d\bm{V}(\bm{e}, \bm{\Theta})}{dt}&= -2 {\bm{e}^\top}{\bm{\mathcal{M}}}{\bm{K}_I}{\bm{\mathcal{M}}^{-1} \bm{e}}+{\bm{e}^\top}\dot{\bm{\mathcal{M}}}{\bm{e}} \\
    &= {\bm{e}^\top} ( -2  {\bm{\mathcal{M}}}{\bm{K}_I}{\bm{\mathcal{M}}^{-1}}+\dot{\bm{\mathcal{M}}})\bm{e},
\end{aligned}
\end{equation}
since $\bm{\mathcal{M}}$ is constant (mass and inertia are time-invariant), we have:
\begin{equation}
\label{eq:Bq_dot}
\begin{aligned}
    \dot{\bm{\mathcal{M}}} = \mathbf{0}_{6 \times 6}.
\end{aligned}
\end{equation}
Substituting into \eqref{eq:Liyapunuofu} yields:
\begin{equation}
\begin{aligned}
\label{eq:KI_condition}
    \frac{d\bm{V}(\bm{e}, \bm{\Theta})}{dt}= 
    {\bm{e}^\top} ( -2 {\bm{\mathcal{M}}}{\bm{K}_I}{\bm{\mathcal{M}}^{-1}} )\bm{e},
\end{aligned}
\end{equation}
since $\bm{\mathcal{M}}$ a positive-definite matrix and $\bm{K}_I$ is a diagonal matrix with positive elements, the matrix $-2 {\bm{\mathcal{M}}}{\bm{K}_I}{\bm{\mathcal{M}}^{-1}}$ is negative definite.
Thus, $\frac{d\bm{V}}{dt} \leq 0$, which ensures Lyapunov stability.
It follows that the observation error $\bm{e}(t)$ asymptotically converges to zero, completing the proof. 
\qedhere 
\end{proof}

Consequently, by substituting \eqref{eq:Lqq_define} into \eqref{eq:we_derivative_from_chenwenhua}, we obtain the observer formulation:
\begin{equation}
\label{eq:observer_matrix_equation}
\begin{aligned}
    \dot{\bm{\hat{d}}_e} = -\bm{K}_I \bm{\mathcal{M}}^{-1} \bm{\hat{d}}_e + \bm{K}_I \bm{\mathcal{M}}^{-1} (\bm{\mathcal{M}} \ddot{\bm{\Theta}} + \bm{\mathcal{C}} \dot{\bm{\Theta}} + \bm{\mathcal{D}} \bm{\mu} + \bm{g}).
\end{aligned}
\end{equation}

Through wind tunnel analysis, we found that while the external force $\bm{f}_e$ exhibits a one-to-one mapping to wind velocity $V_w$ (Sec.~\ref{sec:Force_wind_relationship}), the external torque $\bm{\tau}_e$ does not: the same torque magnitude and direction can correspond to multiple different wind conditions, making torque unsuitable for unambiguous wind reconstruction.
Hence, for the purpose of wind velocity estimation, we extract and discretize the estimated forces part from \eqref{eq:observer_matrix_equation}.
This yields the following discrete-time representation of the force estimator:
\begin{equation}
\label{eq:simple_DOB}
\begin{aligned}
    \bm{\hat{f}}_e(k+1) &= \left( \bm{I}_3 - \frac{\delta t}{2m} \bm{K}_I \right) \bm{\hat{f}}_e(k) \\ &+ \frac{\delta t}{2m} \bm{K}_I \left( m \ddot{\bm{p}} - \bm{g} + u_f \bm{z}_B \right), 
\end{aligned}
\end{equation}
where 
$\delta t$ is the time step;
$\bm{z}_B$ is the third column of $\bm{R}$, representing the body-frame $z$-axis in $\mathcal{F}_I$;
$\bm{\hat{f}}_e$ is the estimated external force in $\mathcal{F}_I$.
The resulting $\bm{\hat{f}}_e$ is transformed into $\mathcal{F}_C$ to yield:
\begin{equation}
    \bm{\hat{f}}_e^C = [\hat{f}_{ex}^C, \hat{f}_{ey}^C, \hat{f}_{ez}^C]^\top.
\end{equation}

\subsubsection{Wind Barrel Aerodynamic Design and Validation}
\label{sec:wind_barrel_design}

To improve wind sensitivity, especially under low-speed conditions, increasing aerodynamic drag is essential to amplify the system’s measurable airflow response.
To this end, we explored the effect of barrel surface texture and developed a customized design that enhances drag.
We evaluated three barrel surface designs, illustrated in Fig.~\ref{fig:wind_barrel_in_wind_tunnel}(a):  
\begin{enumerate}[label=-, leftmargin=*, align=left]
\item 
Design A: cross-lattice texture with large grooves (\SI{6}{mm}$\times$\SI{6}{mm}$\times$\SI{13}{mm});
\item 
Design B~\cite{2011CFL_ReynoldsEffects_Wornom}: smooth surface;
\item 
Design C (ours): same lattice pattern as A but with shallower grooves (\SI{6}{mm}$\times$\SI{6}{mm}$\times$\SI{1.5}{mm}).
\end{enumerate}
All three designs share identical geometry and cross-sectional area, isolating surface structure as the only variable.

\begin{figure}[t]
\vspace{0.0cm}
\centering
\includegraphics[width=0.9\columnwidth]{ 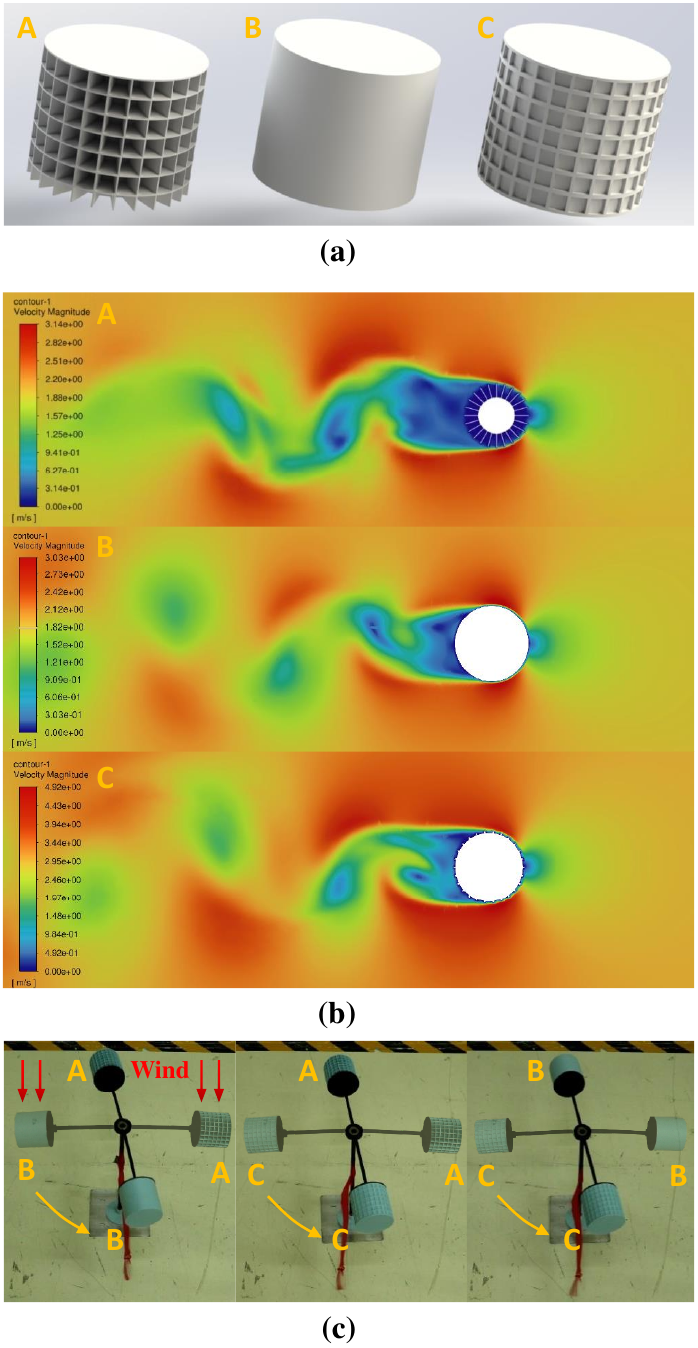}
\caption{
Three wind barrel designs and their evaluation.  
(a) Surface textures of each design.  
(b) CFD simulation results showing airflow and wake patterns.  
(c) Wind tunnel test showing crossbar rotation due to drag differences.
}
\label{fig:wind_barrel_in_wind_tunnel}
\vspace{-0.0cm}
\end{figure}

CFD simulations (Fig.~\ref{fig:wind_barrel_in_wind_tunnel}(b)) reveal distinct wake structures and drag coefficients.
A produces the most symmetric and narrow wake, indicating smooth flow with minimal separation and thus the lowest drag coefficient ($C_d = 0.97$).
In contrast, B exhibits pronounced vortex shedding and a significantly wider wake, reflecting stronger flow separation and a higher drag coefficient ($C_d = 1.30$).
C shows even more intense vortex shedding than B, indicating further increased flow separation, resulting in the highest drag coefficient ($C_d = 1.47$).
These findings confirm that C most effectively increases aerodynamic drag under low-speed flow, making it well-suited for force-based wind sensing.

To experimentally validate the CFD results, we conducted wind tunnel tests using a suspended crossbar setup, as shown in Fig.~\ref{fig:wind_barrel_in_wind_tunnel}(c). 
The crossbar was mounted horizontally and positioned perpendicular to the incoming wind (red arrow), with two different wind barrel designs installed at its ends (e.g., A vs.~B, A vs.~C, B vs.~C), enabling direct comparison of their aerodynamic performance under identical conditions.

The drag force on each barrel is given by
\begin{equation}
\label{eq:resistance_euqation}
F_d = \frac{1}{2} \rho v^2 A C_d,
\end{equation}
where $F_d$ is the drag force, $\rho$ the air density, $v$ the wind speed, $A$ the cross-sectional area, and $C_d$ the drag coefficient. 
Since $A$ is fixed across designs, differences in $F_d$ stem solely from $C_d$.
This imbalance generates a torque that causes the crossbar to rotate, with the higher-drag side moving downwind.

In Fig.~\ref{fig:wind_barrel_in_wind_tunnel}(c), the yellow arrow indicates this rotation.
In the A–B comparison, the crossbar rotated such that B moved downwind, suggesting B generated more drag than A.
In tests involving C, the crossbar consistently rotated with C moving downwind, indicating it produced more drag than both A and B.
These experimental observations confirm the CFD results, which identified C as having the highest $C_d$, validating its advantage in boosting aerodynamic force and improving the signal-to-noise ratio (SNR) in wind sensing.

\subsection{Stage II: Force-to-Wind Mapping}
\label{sec:Force_wind_relationship}

\subsubsection{Wind Tunnel Data Collection}
\label{sec:wind_tunnel_data_collection}

A dedicated dataset was collected in a controlled wind tunnel environment to train the force–to–wind mapping model. 
The UAV was commanded to hover under steady laminar airflow with onboard sensor data recorded continuously at \SI{50}{Hz}.

\paragraph{Horizontal Wind Data}
To characterize the UAV’s aerodynamic response under various horizontal wind conditions, yaw-sweep data were collected at wind speeds ranging from \SI{0}{} to \SI{8}{m/s}, in \SI{1}{m/s} increments.
At each wind speed $V_w$, the UAV maintained a stable hover and was slowly rotated about its yaw axis. 
The yaw angle $\psi$ was incremented by $\ang{10}$ every \SI{20}{\second}, completing a full \ang{360} sweep per speed level.

The horizontal force vector is characterized by its magnitude and direction in $\mathcal{F}_C$:
\begin{equation}
\label{eq:model_input1}
\begin{aligned}
\hat f_h^C &= \sqrt{(\hat{f}_{ex}^C)^2 + (\hat{f}_{ey}^C)^2}, \\
\hat \vartheta_f^C &= \arctan2(\hat{f}_{ex}^C, \hat{f}_{ey}^C),
\end{aligned}
\end{equation}
where $\hat f_h^C$ and $\hat \vartheta_f^C$ are the magnitude and direction of the estimated horizontal external force. 

\begin{figure}[t]
\vspace{-0.0cm}
\centering
\includegraphics[width=1.0\columnwidth]{ 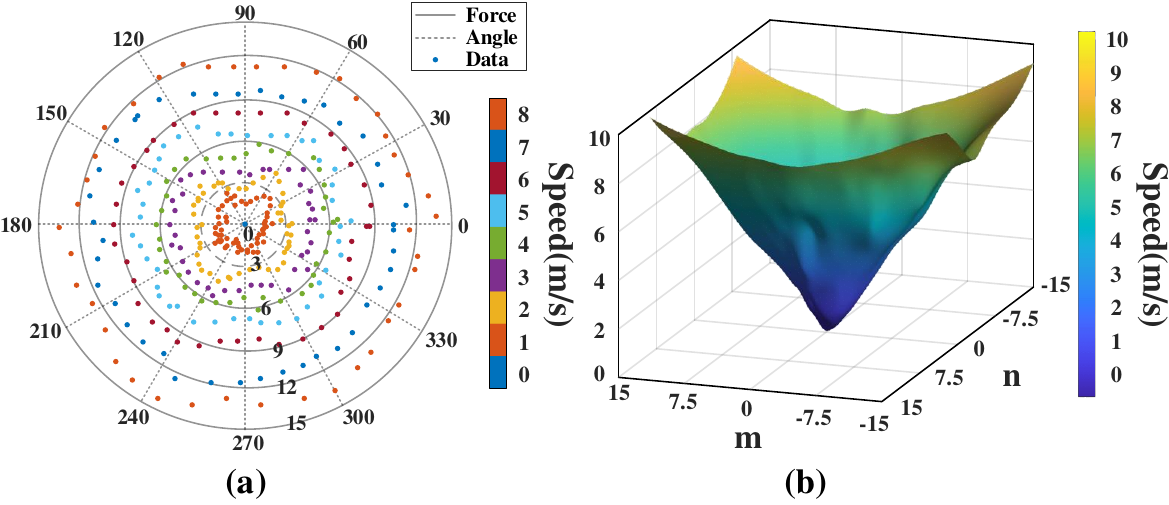}
\caption{
Horizontal force–wind mapping using TPS regression. 
(a) Training data distribution. 
(b) Fitted TPS surface.
}
\label{fig:fitting_model}
\vspace{-0.0 cm}
\end{figure}

Fig.~\ref{fig:fitting_model}(a) visualizes the collected horizontal data in polar coordinates: radial distance represents $\hat f_h^C$, angle denotes $\hat \vartheta_f^C$, and the color gradient encodes the corresponding wind speed. 
The non-circular distribution of points shows the UAV’s anisotropic aerodynamic response across different flow directions, showing that the same wind speed may produce different force magnitudes depending on orientation.

\paragraph{Vertical Wind Data}
Unlike the horizontal direction, vertical wind cannot be easily generated in the wind tunnel, thus we use UAV vertical motion to simulate equivalent vertical airflow. 

To generate an equivalent vertical airflow, the UAV was commanded to perform controlled vertical ascents and descents along $\bm{z}_I$ in still air. 
The UAV was commanded to ascend or descend at constant vertical speeds ranging from \SI{0}{} to \SI{5}{m/s}, increasing in \SI{1}{m/s} increments. 
At each speed level, the motion was sustained for \SI{20}{\second}.
This procedure effectively induces an aerodynamic equivalent of vertical wind, enabling construction of the vertical force–wind model in the absence of actual environmental airflow.

\subsubsection{Force-to-Wind Mapping}
\label{subsec:data_fitting}

To estimate wind velocity from onboard force measurements, we construct a data-driven pipeline that maps estimated external forces to wind components. 
Horizontal wind speed $\hat V_{h}^C$ is obtained by fitting a thin-plate spline model to $(\hat{f}_{ex}^C, \hat{f}_{ey}^C)$, while vertical wind speed $\hat V_{v}^C$ is estimated via polynomial regression on $\hat{f}_{ez}^C$. 
The full wind vector is finally reconstructed in the inertial frame using $\hat V_{h}^C$, $\hat V_{v}^C$, and the estimated direction $\hat \vartheta_f^C$.

\paragraph{TPS Regression of Horizontal Wind Speed}
The mapping between horizontal external force and airflow is modeled using a TPS regression method~\cite{2003JRSS_TPS_Wood}. 

The horizontal force vector, represented in polar form as $(\hat f_h^C, \hat \vartheta_f^C)$, is first converted into Cartesian coordinates:
\begin{equation}
\label{eq:model_input2}
\begin{aligned}
    m &= \hat f_h^C\cos{\hat \vartheta_f^C}, \\
    n &= \hat f_h^C\sin{\hat \vartheta_f^C}.
\end{aligned}
\end{equation}

Using input $(m, n)$, we trained a TPS model $Z(\cdot)$ to predict the corresponding horizontal wind speed $\hat V_{h}^C$:

\begin{equation}
\label{eq:model_input3}
\begin{aligned}
\hat V_{h}^C = Z \left( m,n \right) = Z \left( \hat f_h^C\cos{\hat \vartheta_f^C}, \hat f_h^C\sin{\hat \vartheta_f^C} \right).
\end{aligned}
\end{equation}

The TPS function takes the following form:

\begin{equation}
\label{eq:TPS_model1}
\begin{aligned}
Z(m, n) = \sum_{i=1}^{N} c_i \phi\big(| (m, n) - (m_i, n_i) |\big) + a_0 + a_1 m + a_2 n,
\end{aligned}
\end{equation}
where $\phi(r) = r^2 \log{r}$ is the radial basis function, and $(m_i, n_i)$ is the $i$-th training input. 
The parameters $c_i$, $a_0$, $a_1$, and $a_2$ are fitted by minimizing the regularized loss:

\begin{equation}
\label{eq:TPS_model2}
\begin{aligned}
    \min_{a_0, a_1, a_2, c_i} \left\{ \sum_{i=1}^{N} \left( V_{h,i}^C - Z(m_i, n_i) \right)^2 + \lambda J(Z) \right\},
\end{aligned}
\end{equation}
where $V_{h,i}^C$ is the measured horizontal wind speed in $\mathcal{F}_C$ for the $i$-th sample, $\lambda \ge 0$ is a regularization parameter, and $J(Z)$ is a roughness penalty term controlling surface curvature. 
Further details on TPS modeling and regularization can be found in~\cite{2003JRSS_TPS_Wood}.

Fig.~\ref{fig:fitting_model}(b) shows the fitted TPS surface mapping inputs $(m, n)$ to horizontal wind speed $\hat V_{h}^C$.
The TPS formulation effectively balances fidelity to the training data and smoothness through the regularization term in \eqref{eq:TPS_model2}, allowing the model to capture localized variations in aerodynamic response while suppressing high-frequency noise.
This regularization ensures that the learned force-to-wind relationship faithfully represents underlying physical characteristics without overfitting.

\paragraph{Polynomial Regression of Vertical Wind Speed}
To estimate vertical wind speed $\hat V_{v}^C$, a simple polynomial regression model was fitted to relate the estimated vertical external force $\hat{f}_{ez}^C$ to $\hat V_{v}^C$:
\begin{equation}
\label{eq:wind_force_z}
\begin{aligned}
    \hat V_{v}^C = \sum_{k=1}^{K} c_k \cdot (\hat{f}_{ez}^C)^k,
\end{aligned}
\end{equation}
where 
$K$ is the polynomial degree, and $c_k$ are the regression coefficients obtained during training.

\paragraph{Wind Vector Reconstruction}
The estimated values $\hat \vartheta_f^C$, $\hat V_{h}^C$, and $\hat V_{v}^C$ from \eqref{eq:model_input1}, \eqref{eq:model_input3}, and \eqref{eq:wind_force_z}, 
are used to construct the relative air velocity $\hat{\bm{A}}_r^C$ in $\mathcal{F}_C$:
\begin{equation}
\label{eq:construct_Arc}
\hat{\bm{A}}_r^C =
\begin{bmatrix}
\hat V_{h}^C \cos\hat \vartheta_f^C \\
\hat V_{h}^C \sin\hat \vartheta_f^C \\
\hat V_{v}^C
\end{bmatrix},
\end{equation}

This vector is then rotated into $\mathcal{F}_I$ as $\hat{\bm{A}}_r^I$ via:
\begin{equation}
\label{eq:fitting_final}
\hat{\bm{A}}_r^I = \mathbf{R}_{IC} \hat{\bm{A}}_r^C,
\end{equation}
where 
$\mathbf{R}_{IC}$ is the rotation matrix from $\mathcal{F}_C$ to $\mathcal{F}_I$. 

According to the wind triangle principle, the true wind velocity $\hat{\bm{A}}_w^I$ in $\mathcal{F}_I$ is computed by subtracting the UAV's ground velocity $\dot{\bm{p}}$ from the air-relative velocity $\hat{\bm{A}}_r^I$:
\begin{equation}
\label{eq:wind_vector}
\hat{\bm{A}}_w^I = \hat{\bm{A}}_r^I - \dot{\bm{p}}.
\end{equation}

Given the wind vector $\hat{\bm{A}}_w^I = [\hat A_{wx}, \hat A_{wy}, \hat A_{wz}]^\top$ in $\mathcal{F}_I$, the horizontal wind direction $\hat \vartheta_h^I$, speed $\hat V_{h}^I$, and vertical speed $\hat V_{v}^I$ are computed as:
\begin{equation}
\label{eq:dynamic_fitting_equ}
\begin{aligned}
\hat \vartheta_{h}^I &= \arctan2(\hat A_{wy}, \hat A_{wx}) + \pi, \\
\hat V_{h}^I &= \sqrt{\hat A_{wx}^2 + \hat A_{wy}^2}, \\
\hat V_{v}^I &= \hat A_{wz},
\end{aligned}
\end{equation}
where the additional $\pi$ compensates for the \ang{180} offset between the force orientation and the true wind direction.
A positive value of $\hat V_{v}^I$ indicates upward airflow, while a negative value indicates downward airflow along $\bm{z}_I$.
This completes the reconstruction of the full wind velocity in $\mathcal{F}_I$.


\begin{table*}[t]
\centering
\caption{
The mean ($e_\mu$) and SD ($e_\sigma$) of DOB-estimated force errors under different ground-truth force magnitudes applied along the $\bm{x}_I$, $\bm{y}_I$, and $\bm{z}_I$ axes.
}
\label{tab:dob_force_levels}
\fontsize{12}{13}\selectfont
\begin{tabularx}{\textwidth}{>{\centering\arraybackslash}m{3cm} | *{6}{>{\centering\arraybackslash}X}}
\specialrule{0.2em}{0em}{0.2em}
\textbf{GT Force (N)} 
  & \multicolumn{2}{c}{4.89} 
  & \multicolumn{2}{c}{9.78} 
  & \multicolumn{2}{c}{12.71} \\
\cmidrule(lr){2-3} \cmidrule(lr){4-5} \cmidrule(lr){6-7}
\textbf{Axis} 
  & $e_\mu$ (N) & $e_\sigma$ (N)
  & $e_\mu$ (N) & $e_\sigma$ (N)
  & $e_\mu$ (N) & $e_\sigma$ (N) \\
\midrule
$\bm{x}_I$-direction 
  & -0.27 & 0.09 
  & -0.25 & 0.10 
  & -0.20 & 0.07 \\
$\bm{y}_I$-direction 
  & -0.72 & 0.16 
  & -0.71 & 0.15 
  & -0.73 & 0.13 \\
$\bm{z}_I$-direction 
  & -1.88 & 0.33 
  & -1.85 & 0.22 
  & -1.87 & 0.19 \\
\bottomrule
\end{tabularx}
\end{table*}

\begin{table*}[ht]
\centering
\caption{
Wind tunnel validation: 
RMSE of estimated wind speed $\varepsilon_v$ (\si{m/s}) and direction $\varepsilon_\theta$ (\si{\degree}) at selected ground-truth wind speeds.
}
\label{tab:wind_tunnel_big_table}
\renewcommand{\arraystretch}{1.5}
\begin{tabularx}{\textwidth}{>{\centering\arraybackslash}m{3cm} | *{10}{>{\centering\arraybackslash}X}}
\specialrule{0.2em}{0em}{0.2em}
\textbf{GT Wind Speed (\si{m/s})} 
  & \multicolumn{2}{c}{0} & \multicolumn{2}{c}{1} 
  & \multicolumn{2}{c}{5} & \multicolumn{2}{c}{9} 
  & \multicolumn{2}{c}{10}  \\
\cmidrule(lr){2-3}\cmidrule(lr){4-5}\cmidrule(lr){6-7}\cmidrule(lr){8-9}\cmidrule(lr){10-11}
\textbf{Method} 
  & $\varepsilon_v$ (\si{m/s}) & $\varepsilon_\theta$ (\si{\degree})
  & $\varepsilon_v$ (\si{m/s}) & $\varepsilon_\theta$ (\si{\degree})
  & $\varepsilon_v$ (\si{m/s}) & $\varepsilon_\theta$ (\si{\degree})
  & $\varepsilon_v$ (\si{m/s}) & $\varepsilon_\theta$ (\si{\degree})
  & $\varepsilon_v$ (\si{m/s}) & $\varepsilon_\theta$ (\si{\degree}) \\
\midrule
Zimmerman~\textit{et al.}~\cite{2022Measurement_WindML1_Zimmerman}
  & 0.62 & 13.2 & 0.44 & 12.6 & 0.17 & 9.4 & 0.16 & 8.9 & 0.17 & 7.8 \\
Yu~\textit{et al.}~\cite{2024IROS_DOB_Yu}
  & 0.48 & 11.0 & 0.39 & 8.4 & 0.15 & 5.1 & 0.15 & 5.0 & 0.17 & 5.1 \\
Proposed
  & \textbf{0.32} & \textbf{9.8} & \textbf{0.17} & \textbf{7.6} & \textbf{0.08} & \textbf{5.0} & \textbf{0.06} & \textbf{3.7} & \textbf{0.06} & \textbf{3.6} \\
\bottomrule
\end{tabularx}
\end{table*}

\subsubsection{Dynamic Filtering for Wind Estimate Refinement}
\label{sec:dynamice_filtering}

Although the wind vector is reconstructed via force-to-wind mapping, it remains sensitive to high-frequency noise, primarily due to sensor noise and numerical differentiation in DOB estimation. 
Therefore, we apply temporal filtering to refine the signal quality.
According to~\cite{2020TMECH_Generalized_DOBVtol_Zheng}, the DOB maintains accurate gain up to approximately \SI{0.5}{\hertz}, sufficient for low-altitude wind monitoring tasks where rapid response is less critical.
To enhance real-time usability, we apply a dynamic post-processing filter. 
This filter uses gain scheduling to suppress high-frequency noise while preserving temporal responsiveness.
Specifically, the filter gain adapts to the estimated wind speed: higher gains are applied under low-speed conditions to improve the SNR, while lower gains are used at higher speeds to minimize phase lag.

Fig.~\ref{fig:dynamic_filter} compares three approaches: unfiltered wind estimation, conventional static filtering, and the proposed dynamic filter. 
The proposed dynamic filter (blue curve) effectively suppresses high-frequency noise without introducing noticeable delay. 
At a wind speed of \SI{10}{m/s}, it reduces lag by \SI{3.4}{\second} compared to static filtering, as highlighted by the black arrow. 
Meanwhile, it maintains close temporal alignment with the unfiltered signal while offering substantial noise suppression.
These results confirm that dynamic filtering improves both the clarity and timeliness of wind estimates, making them more suitable for practical applications.

\begin{figure}[t]
\centering
\includegraphics[width=1.0\columnwidth]{ 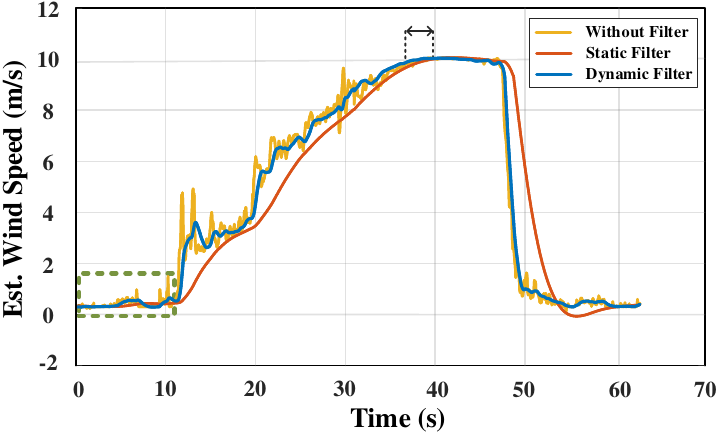}
\caption{
Comparison of wind estimation using unfiltered, static-filtered, and dynamic-filtered methods.
The dynamic filter effectively reduces noise and remains responsive.
}
\label{fig:dynamic_filter}
\vspace{-0.0cm}
\end{figure}


\section{EXPERIMENT AND RESULTS}
\label{sec:Experiment}
\subsection{Hardware Platform}
The DJI Matrice 300 UAV (see Fig.~\ref{fig:m300_combine_coord}) was used as the experimental platform.
To improve wind sensitivity and estimation accuracy, a custom wind barrel was mounted beneath the UAV.
The wind barrel is a cylindrical structure with an inner diameter of \SI{105}{mm}, outer diameter of \SI{106.5}{mm}, height of \SI{150}{mm}, and mass of \SI{105}{g}.
The barrel is installed in an underslung configuration to enhance aerodynamic response without increasing the vehicle's geometric.
The performance of the wind barrel is described in Sec.~\ref{sec:wind_barrel_design}.
The modified UAV is capable of reaching a maximum flight speed of \SI{15}{m/s}. 

\subsection{Validation of DOB-Based Force Estimation}

To validate the performance of the DOB-based external force estimation, we conducted static experiments using a tensiometer. 
Ground-truth (GT) forces of \SI{4.89}{N}, \SI{9.78}{N}, and \SI{12.71}{N} were sequentially applied along the $\bm{x}_I$, $\bm{y}_I$, and $\bm{z}_I$ axes.
Table~\ref{tab:dob_force_levels} summarizes the mean ($e_\mu$) and standard deviation (SD) ($e_\sigma$) of DOB force estimation errors for each axis and force magnitude.
For the $\bm{x}_I$-axis, $e_\mu$ remained small, ranging from \SI{-0.27}{N} to \SI{-0.20}{N}, with a low $e_\sigma$ between \SI{0.07}{N} and \SI{0.10}{N}.
For the $\bm{y}_I$-axis, the DOB also exhibited a consistent negative bias of around \SI{0.72}{N}, with $e_\sigma$ ranging from \SI{0.13}{N} to \SI{0.16}{N}.
The $\bm{z}_I$-axis demonstrated the largest $e_\mu$, consistently around \SI{-1.87}{N}. 
Despite the larger bias, the $e_\sigma$ decreased as the applied force increased, from \SI{0.33}{N} at \SI{4.89}{N} to \SI{0.19}{N} at \SI{12.71}{N}, indicating improved stability under stronger forces.

The estimation errors were primarily due to sensor noise. 
Across all axes and force magnitudes, the low $e_{\sigma}$ indicates consistent and reliable DOB outputs. 
Although a systematic bias exists in $e_\mu$, it does not compromise wind estimation accuracy. 
This is because there exists a one-to-one mapping between the external force $\bm{f}_e^C$ and the wind vector $\hat{\bm{A}}_w^I$. 
As long as the DOB consistently estimates $\bm{f}_e^C$ as $\bm{\hat{f}}_e^C$, the system can learn a reliable mapping from $\bm{\hat{f}}_e^C$ to $\hat{\bm{A}}_w^I$. 
The small $e_\sigma$ ensures that this mapping remains stable and accurate despite the bias. 
Therefore, the proposed DOB-based force estimation is robust to sensor noise and adaptable to different UAV platforms via parameter tuning.


\subsection{Wind Estimation: Experimental Setup}
To systematically evaluate our wind estimation method, we conduct experiments in a wind tunnel, indoor space, and outdoor environment. 

\paragraph{Ground-Truth Wind Measurement}
In wind-tunnel experiments, the horizontal flow was set and controlled by the tunnel, while vertical wind components were not measured.
Indoor experiments assumed still air, with a ground-truth wind speed of \SI{0}{m/s} and direction of \ang{0} (equivalently \ang{360}).
For outdoor experiments, a 2D ultrasonic anemometer placed \SI{6}{m} from the UAV and sampled at \SI{0.3}{Hz} provided horizontal wind GT.
Vertical wind estimation was evaluated only indoors due to the lack of reliable vertical sensing for the wind-tunnel and outdoor setups.

\paragraph{Baselines}
We compared the proposed method with a regularized polynomial regression method (Yu~\textit{et al.}~\cite{2024IROS_DOB_Yu}) and a learning-based method (Zimmerman~\textit{et al.}~\cite{2022Measurement_WindML1_Zimmerman}). 
Both baselines were re-implemented on our UAV platform using consistent sensor configurations.  
Since these baselines estimate only horizontal wind, comparisons were limited to horizontal wind estimation.
The learning-based method was trained using 150,000 samples with 11 input features, including motor speed ($\Omega$), attitude ($\bm{\eta}$), and acceleration ($\bm{\ddot{p}}$). 
In contrast, our TPS model used only 289 averaged samples with inputs $\left( m,n \right)$ and $V_{h}^C$ as defined in Sec.~\ref{sec:Force_wind_relationship}.

\paragraph{Evaluation Metrics}
Wind estimation performance was quantified using two primary metrics: root-mean-square error (RMSE) and the Pearson correlation coefficient ($r$), evaluated for both wind speed and direction relative to ground truth.
RMSE reflects absolute estimation accuracy, while $r$ captures the linear correlation between estimated and true values.
Additionally, the mean and standard deviation (SD) of RMSE across multiple samples were computed to assess average performance and consistency.

\subsection{Wind Estimation: Wind Tunnel Validation}

Wind tunnel experiments with controlled horizontal winds were conducted to assess the accuracy, robustness, and generalization of the proposed method.
The wind speed was gradually increased from \SI{0}{m/s} to \SI{10}{m/s} in \SI{1}{m/s} increments. 
At each speed level, the UAV was held stationary at a fixed angle to the wind for \SI{20}{\second}. 
The test data at \SI{9}{m/s} and \SI{10}{m/s} beyond the training range of \SIrange{0}{8}{m/s} were used to assess model extrapolation.

Fig.~\ref{fig:wind_tunnel_temp_points} shows the RMSE and SD of estimated horizontal wind speed and direction across wind speeds from \SI{0}{} to \SI{10}{m/s}.
The white region indicates interpolation within the training range (\SI{0}{}–\SI{8}{m/s}), while the gray area denotes extrapolation (\SI{9}{}–\SI{10}{m/s}).
For wind speed (top panel), all methods exhibited decreasing RMSE and SD with increasing wind speed, stabilizing above \SI{3}{m/s}. 
This trend reflects changes in SNR: at low speeds, weaker aerodynamic forces reduce SNR, amplifying sensor noise; at higher speeds, SNR improves, reducing estimation errors.
For wind direction (bottom panel), the proposed method and Yu~\textit{et al.}~\cite{2024IROS_DOB_Yu} achieved comparable RMSE across the full range.
Overall, the proposed method consistently achieved the lowest RMSE and SD, outperforming all baselines in both interpolation and extrapolation regimes.

\begin{figure*}[t]
\centering
\includegraphics[width=1.9\columnwidth]{ 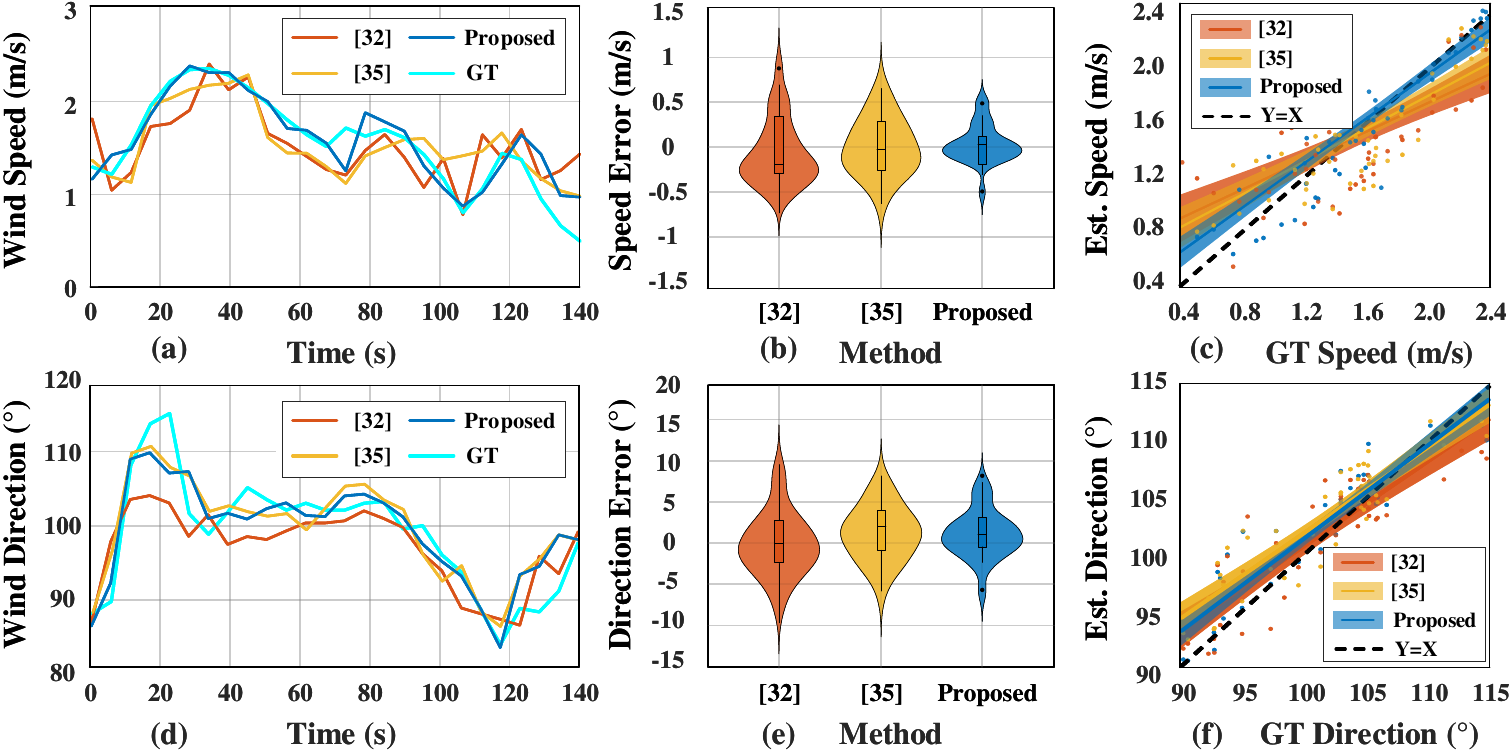}
\caption{ 
Outdoor hover.  
(a) Wind speed estimates over time.
(b) Estimated wind speed error distribution.
(c) Correlation between GT speed and estimated speed.
(d) Wind direction estimates over time.
(e) Estimated wind direction error distribution.
(f) Correlation between GT direction and estimated direction.
}
\label{fig:outdoor_hover}
\end{figure*}

Table~\ref{tab:wind_tunnel_big_table} reports the RMSE of each method at selected wind speeds. The proposed method consistently achieves the lowest estimation error across all tested wind speeds, with up to \SI{65}{\%} reduction in speed RMSE and \SI{54}{\%}  reduction in direction RMSE compared to the strongest baseline. 
Notably, even at the highest speed (\SI{10}{m/s})—beyond the training range—the proposed method maintains superior accuracy, confirming robust extrapolation capability.

Compared to the global polynomial fit in Yu~\textit{et al.}~\cite{2024IROS_DOB_Yu}, our TPS model better preserves force-to-wind mapping features, improving accuracy.
The end-to-end learning model in Zimmerman~\textit{et al.}~\cite{2022Measurement_WindML1_Zimmerman} directly maps sensor data to wind estimates but performs worse in our setup---likely due to the lack of physical constraints, which makes it prone to overfitting and less physically consistent.
In contrast, our method explicitly incorporates the physical relationship between wind and aerodynamic force, enhancing interpretability and consistency.
Overall, the proposed method demonstrates superior performance and generalization under controlled laminar flow.

\begin{figure}[t]
\vspace{-0.0cm}
\centering
\includegraphics[width=1.0\columnwidth]{ 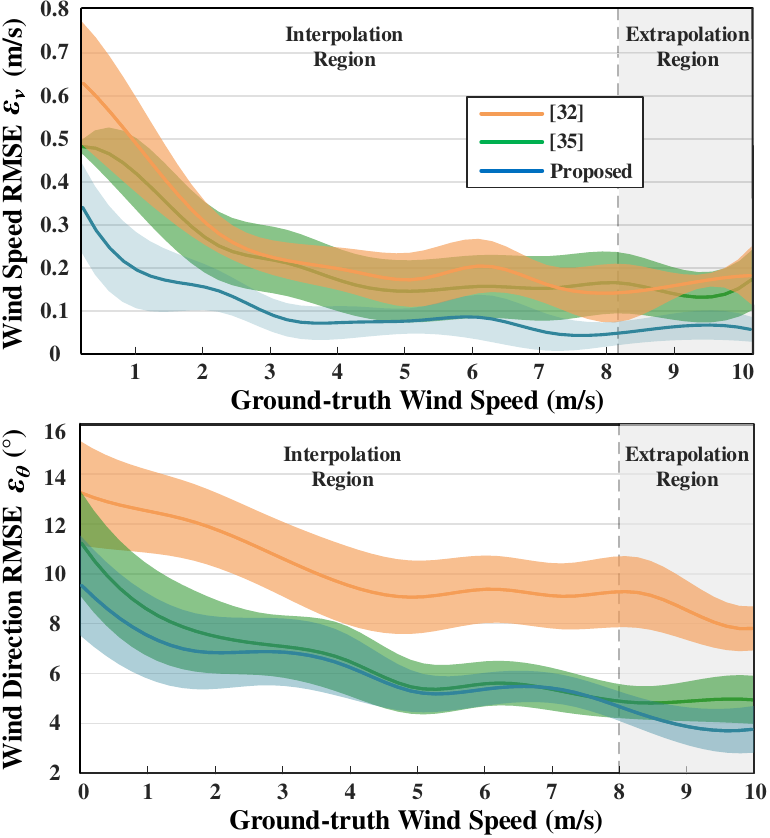}
\caption{
Wind tunnel validation: 
RMSE (solid lines) and SD (shaded areas) of estimated horizontal wind speed and direction over ground-truth wind speeds from \SI{0}{} to \SI{10}{m/s}.
}
\label{fig:wind_tunnel_temp_points}
\vspace{-0.0cm}
\end{figure}

\begin{figure*}[t]
\vspace{-0.0cm}
\centering
\includegraphics[width=1.95\columnwidth]{ 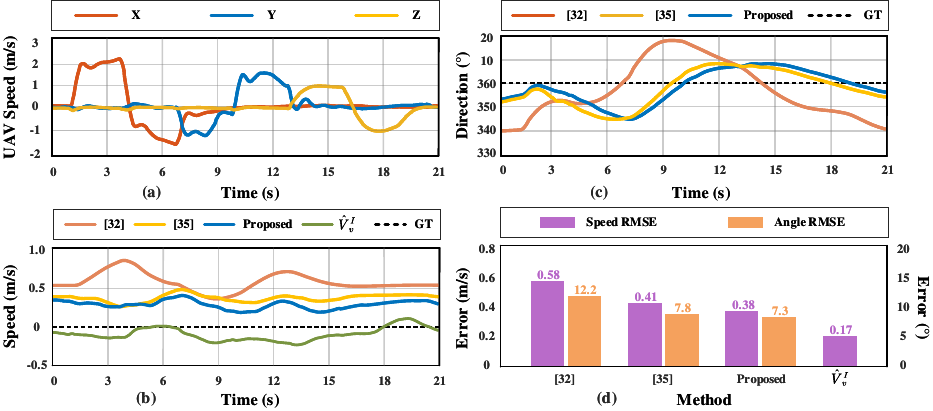}
\caption{
Indoor dynamic flight: fast lateral translation.
(a) UAV speeds along $\bm{x}_I$, $\bm{y}_I$, and $\bm{z}_I$.
(b) Estimated wind speeds $\hat V_{h}^I$ (proposed) and $\hat V_{v}^I$.
(c) Estimated horizontal wind direction $\hat \vartheta_{h}^I$.
(d) RMSE of wind speeds and horizontal wind direction across all methods.
}
\label{fig:IndoorDynamic_FastLateralTranslatione}
\vspace{-0.0cm}
\end{figure*}

\subsection{Wind Estimation: Field Validation}

\subsubsection{Outdoor Hover}
\label{subsubsec:Outdoor_Hover}

To assess estimation accuracy under static, real-world conditions, an outdoor hovering experiment was conducted, with results summarized in Fig.~\ref{fig:outdoor_hover}.

As shown in Fig.~\ref{fig:outdoor_hover}(a), the proposed method (blue) closely matches the ground truth (cyan), outperforming both Zimmerman~\textit{et al.}~\cite{2022Measurement_WindML1_Zimmerman} (brownish red) and Yu~\textit{et al.}~\cite{2024IROS_DOB_Yu} (yellow). 
The proposed method achieved a wind speed RMSE of \SI{0.22}{m/s}, representing relative improvements of \SI{41}{\%} and \SI{27}{\%} compared to Zimmerman~\textit{et al.}~\cite{2022Measurement_WindML1_Zimmerman} (\SI{0.37}{m/s}) and Yu~\textit{et al.}~\cite{2024IROS_DOB_Yu} (\SI{0.30}{m/s}), respectively.
The wind speed error distributions (Fig.~\ref{fig:outdoor_hover}(b), left) demonstrate the proposed method's superior performance with a sharp, zero-centered peak, in contrast to the wider error ranges of the other methods.
$r$ between estimated and ground-truth wind speeds (Fig.~\ref{fig:outdoor_hover}(c)) was $0.94$ for the proposed method, surpassing Zimmerman~\textit{et al.}~\cite{2022Measurement_WindML1_Zimmerman} ($r = 0.75$) and Yu~\textit{et al.}~\cite{2024IROS_DOB_Yu} ($r = 0.83$).
Fig.~\ref{fig:outdoor_hover}(d) illustrates the estimated wind direction during hover.
The proposed method achieved an RMSE of \ang{3.3}, which is \SI{52}{\%} lower than Zimmerman~\textit{et al.}~\cite{2022Measurement_WindML1_Zimmerman} (\ang{6.9}) and \SI{13}{\%} lower than Yu~\textit{et al.}~\cite{2024IROS_DOB_Yu} (\ang{3.8}).
Fig.~\ref{fig:outdoor_hover}(e) presents the distribution of wind direction estimation errors.
As illustrated in Fig.~\ref{fig:outdoor_hover}(f), the proposed method achieved $r=0.93$, outperforming Zimmerman~\textit{et al.}~\cite{2022Measurement_WindML1_Zimmerman} ($r=0.87$) and Yu~\textit{et al.}~\cite{2024IROS_DOB_Yu} ($r=0.91$).

In summary, the proposed approach demonstrates superior accuracy in wind vector estimation under static outdoor conditions, exhibiting both low estimation error and high consistency with ground truth.

\subsubsection{Indoor Dynamic Flight}

To assess wind estimation performance in still air, two indoor flight patterns were tested:  
(1) a horizontal figure-eight for low-speed, continuously changing headings, and  
(2) a fast lateral translation for high-speed motion along all principal axes.

\begin{figure}[t]
\centering
\includegraphics[width=0.9\columnwidth]{ 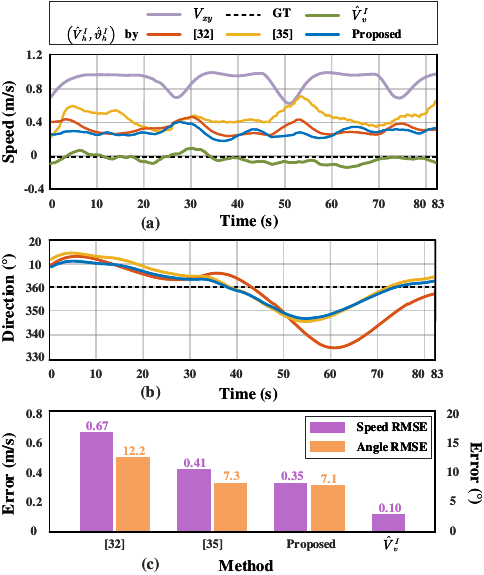}
\caption{
Indoor dynamic flight: horizontal figure-eight pattern. 
(a) Horizontal UAV speed $V_{xy}$, estimated wind speeds $\hat V_{h}^I$ (horizontal) and $\hat V_{v}^I$ (vertical).
(b) Estimated horizontal wind direction $\hat \vartheta_h^I$. 
(c) RMSE of wind speeds and horizontal wind direction (vertical wind speed RMSE shown only for the proposed method).
}
\label{fig:IndoorDynamic_HorizontalFigureEight}
\vspace{-0.0cm}
\end{figure}

\paragraph{Horizontal Figure-Eight Flight}
To evaluate performance during low-speed maneuvers with varying headings, the UAV followed a horizontal figure-eight trajectory.

Fig.~\ref{fig:IndoorDynamic_HorizontalFigureEight}(a) shows the UAV’s horizontal speed $V_{xy}$, estimated horizontal wind speeds $\hat V_{h}^I$ from Zimmerman~\textit{et al.}~\cite{2022Measurement_WindML1_Zimmerman}, Yu~\textit{et al.}~\cite{2024IROS_DOB_Yu}, and the proposed method, along with the proposed method’s vertical wind speed estimate $\hat V_{z}^I$. 
The proposed method produced horizontal wind estimates closest to the zero-wind ground truth, with reduced bias and fluctuations, and uniquely provided vertical wind estimation. 
Fig.~\ref{fig:IndoorDynamic_HorizontalFigureEight}(b) shows the estimated horizontal wind direction $\hat \vartheta_{h}^I$, where the proposed method exhibited the lowest deviation.
Quantitatively (Fig.~\ref{fig:IndoorDynamic_HorizontalFigureEight}(c)), the proposed method’s horizontal wind speed RMSE was \SI{0.35}{m/s}, \SI{48}{\%} lower than Zimmerman~\textit{et al.}~\cite{2022Measurement_WindML1_Zimmerman} (\SI{0.67}{m/s}) and \SI{15}{\%} lower than Yu~\textit{et al.}~\cite{2024IROS_DOB_Yu} (\SI{0.41}{m/s}).
Its vertical wind speed RMSE was \SI{0.1}{m/s}.  
For horizontal wind direction, the RMSE was \ang{7.1}, \SI{42}{\%} lower than Zimmerman~\textit{et al.}~\cite{2022Measurement_WindML1_Zimmerman} (\ang{12.2}) and \SI{3}{\%} lower than Yu~\textit{et al.}~\cite{2024IROS_DOB_Yu} (\ang{7.3}).

Errors were mainly due to thrust modeling inaccuracies during dynamic flight, particularly from variations in propeller advance ratio affecting aerodynamic thrust~\cite{2011AIAAASM_Propeller_Brandt}.

\paragraph{Fast Lateral Translation Flight}
To assess performance during rapid translations along $\bm{x}_I$, $\bm{y}_I$, and $\bm{z}_I$ axes, the UAV executed acceleration–deceleration maneuvers along each axis to cover both horizontal and vertical motions.

Fig.~\ref{fig:IndoorDynamic_FastLateralTranslatione}(a) shows the UAV’s translational speeds along each axis. 
Fig.~\ref{fig:IndoorDynamic_FastLateralTranslatione}(b) presents horizontal wind speed estimates $\hat V_{h}^I$ from all three methods and vertical wind speed $\hat V_{v}^I$ from the proposed method.  
Fig.~\ref{fig:IndoorDynamic_FastLateralTranslatione}(c) shows estimated wind direction $\hat\vartheta_{h}$, and Fig.~\ref{fig:IndoorDynamic_FastLateralTranslatione}(d) summarizes RMSEs.

For horizontal wind speed, the proposed method achieved an RMSE of \SI{0.38}{m/s}, \SI{34}{\%} lower than Zimmerman~\textit{et al.}~\cite{2022Measurement_WindML1_Zimmerman} (\SI{0.58}{m/s}) and \SI{7}{\%} lower than Yu~\textit{et al.}~\cite{2024IROS_DOB_Yu} (\SI{0.41}{m/s}). 
For wind direction, the RMSE was \ang{7.3}, \SI{40}{\%} lower than Zimmerman~\textit{et al.}~\cite{2022Measurement_WindML1_Zimmerman} (\ang{12.2}) and \SI{6}{\%} lower than Yu~\textit{et al.}~\cite{2024IROS_DOB_Yu} (\ang{7.8}). 
The proposed method’s vertical wind speed RMSE was \SI{0.17}{m/s}, with additional errors from rotor-induced downwash in the confined indoor space.

\paragraph{In summary}
Across both experiments, the proposed method consistently outperformed the other two, achieving the lowest RMSE in horizontal wind speed and direction, while maintaining low vertical wind error in dynamic indoor flight.

\subsubsection{Outdoor Dynamic Flight}
To evaluate wind estimation performance under realistic, time-varying conditions, outdoor dynamic flight experiments were conducted in the presence of the natural wind.

\paragraph{Circular Flight}
The UAV followed circular trajectories to induce varying airflow across all axes.

Fig.~\ref{fig:OutdoorDynamic_Circular} shows the estimated horizontal wind speed $\hat V_{h}^I$ and direction $\hat \vartheta_{h}^I$ along with ground truth from the ultrasonic anemometer. 
The proposed method achieved the highest accuracy, while baselines exhibited larger deviations. 

Corresponding RMSE and Pearson correlation coefficient ($r$) are summarized in Table~\ref{tab:OutdoorDynamic_Circular}.
For wind speed estimation, Zimmerman~\textit{et al.}~\cite{2022Measurement_WindML1_Zimmerman} reported an RMSE of \SI{0.55}{m/s} with $r=\SI{82}{\%}$.
Yu~\textit{et al.}~\cite{2024IROS_DOB_Yu} improved these results to \SI{0.49}{m/s} and \SI{85}{\%}.
The proposed method further reduced the error to \SI{0.29}{m/s} with $r = \SI{93}{\%}$, representing RMSE reductions of \SI{47}{\%} and \SI{41}{\%} compared to Zimmerman~\textit{et al.}~\cite{2022Measurement_WindML1_Zimmerman} and Yu~\textit{et al.}~\cite{2024IROS_DOB_Yu}, respectively.
For wind direction estimation, Zimmerman~\textit{et al.}~\cite{2022Measurement_WindML1_Zimmerman} reported an RMSE of \ang{14.7} with $r = \SI{76}{\%}$.
Yu~\textit{et al.}~\cite{2024IROS_DOB_Yu} improved the estimation with a lower RMSE of \ang{6.4} and a higher $r$ of \SI{83}{\%}.
The proposed method achieved the lowest RMSE (\ang{6.0}) with $r = \SI{85}{\%}$, corresponding to a further \SI{59}{\%} over Zimmerman~\textit{et al.}~\cite{2022Measurement_WindML1_Zimmerman} and \SI{6}{\%} reduction over Yu~\textit{et al.}~\cite{2024IROS_DOB_Yu}.
Overall, the proposed method consistently outperformed the baselines in horizontal wind estimation, achieving both the highest accuracy and the strongest correlation with ground truth.

\begin{table}[h]
\centering
\caption{
Outdoor dynamic flight: RMSE and Pearson correlation ($r$) of estimated wind speed and direction.
}
\label{tab:OutdoorDynamic_Circular}
\renewcommand{\arraystretch}{1.0}
\setlength{\tabcolsep}{2pt}
\begin{tabular}{
c
S[table-format=1.2]
S[table-format=2.0]
S[table-format=2.1]
S[table-format=2.0]
}
\toprule
  & \multicolumn{2}{c}{\textbf{Wind Speed} (\si{m/s})} 
  & \multicolumn{2}{c}{\textbf{Wind Direction} (\si{\degree})} \\
\cmidrule(lr){2-3} \cmidrule(lr){4-5}
\textbf{Method} 
& \multicolumn{1}{c}{RMSE $\downarrow$} & \multicolumn{1}{c}{$r$ (\%) $\uparrow$} & \multicolumn{1}{c}{RMSE $\downarrow$} & \multicolumn{1}{c}{$r$ (\%) $\uparrow$} \\
\midrule
Zimmerman~\textit{et al.}~\cite{2022Measurement_WindML1_Zimmerman}
  & 0.55 & 82 & 14.7 & 76 \\
Yu~\textit{et al.}~\cite{2024IROS_DOB_Yu}
  & 0.49 & 85 & 6.4 & 83 \\
Proposed
  & \bfseries{0.29} & \bfseries{93} & \bfseries{6.0} & \bfseries{85} \\
\bottomrule
\end{tabular}
\end{table}

\begin{figure}[t]
\vspace{0.0cm}
\centering
\includegraphics[width=1.0\columnwidth]{ 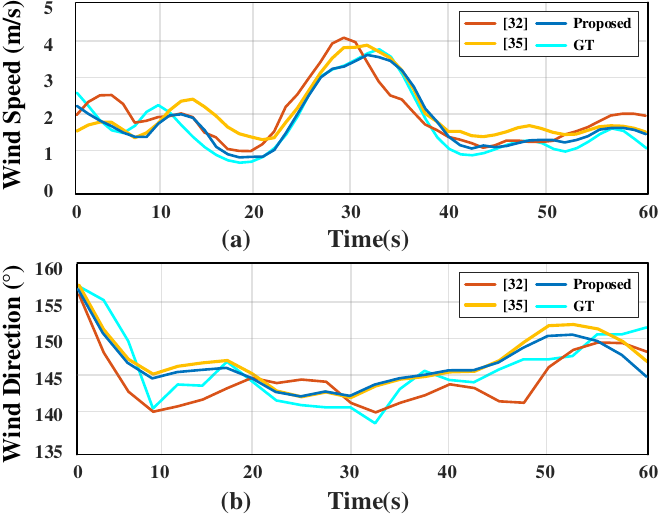}
\caption{
Outdoor dynamic flight: circular pattern. 
(a) Estimated horizontal wind speed $\hat V_{h}^I$ vs. ground truth. 
(b) Estimated horizontal wind direction $\hat \vartheta_{h}^I$ vs. ground truth.
}
\label{fig:OutdoorDynamic_Circular}
\vspace{-0.0cm}
\end{figure}

\paragraph{Extended Free Flights}

\begin{table*}[t]
\centering
\caption{
Outdoor dynamic flight: extended free flights.
RMSE of estimated wind speed $\varepsilon_v$ (\si{m/s}) and direction $\varepsilon_\theta$ (\si{\degree}) for multiple trials of hover and dynamic flights.
}
\label{tab:OutdoorDynamic_ExtendedFreeFlight}
\fontsize{9}{9}\selectfont 
\begin{tabularx}{\textwidth}{>{\raggedright\arraybackslash}m{3cm} | *{6}{>{\centering\arraybackslash}X}}
\specialrule{0.2em}{0em}{0.2em}
\textbf{Method} 
  & \multicolumn{2}{c}{Zimmerman~\textit{et al.}~\cite{2022Measurement_WindML1_Zimmerman}} & \multicolumn{2}{c}{Yu~\textit{et al.}~\cite{2024IROS_DOB_Yu}} & \multicolumn{2}{c}{Proposed} \\
\cmidrule(lr){2-3}\cmidrule(lr){4-5}\cmidrule(lr){6-7}
\textbf{Flights} 
  & $\varepsilon_v$ (\si{m/s}) & $\varepsilon_\theta$ (\si{\degree})
  & $\varepsilon_v$ (\si{m/s}) & $\varepsilon_\theta$ (\si{\degree}) 
  & $\varepsilon_v$ (\si{m/s}) & $\varepsilon_\theta$ (\si{\degree}) \\
\midrule
Trial 1 (Hover) & 0.58 & 8.7 & 0.35 & 7.0 & \textbf{0.28} & \textbf{6.9} \\
Trial 2 (Hover) & 0.40 & 6.9 & 0.27 & 5.5 & \textbf{0.23} & \textbf{5.4} \\
Trial 3 (Dynamic) & 0.70 & 11.8 & 0.50 & 10.1 & \textbf{0.46} & \textbf{8.9} \\
Trial 4 (Dynamic) & 0.54 & 13.1 & 0.44 & 11.4 & \textbf{0.42} & \textbf{10.5} \\
Trial 5 (Dynamic) & 0.70 & 13.5 & 0.48 & 7.9 & \textbf{0.37} & \textbf{7.1} \\
\bottomrule
\end{tabularx}
\end{table*}

To further assess performance in less constrained scenarios, extended outdoor free-flight tests were conducted in naturally varying wind, including both hover and dynamic maneuvers.
The UAV flew \SI{5}{\meter}–\SI{10}{\meter} from the ultrasonic anemometer following unconstrained trajectories. 
Five trials, each lasting up to 6~minutes, were recorded.
Trials~1–2 (Hover), conducted under the same setup as Sec.~\ref{subsubsec:Outdoor_Hover}, were used as reference under steady conditions for comparison with dynamic trials.
Trials~3–5 (Dynamic) involved free flights with heading and velocity changes.

Table~\ref{tab:OutdoorDynamic_ExtendedFreeFlight} reports the RMSE of estimated wind speed $\varepsilon_v$ and direction $\varepsilon_\theta$ for each trial and method.
Hover trials exhibited relatively low errors for all methods, with the proposed method achieving the smallest RMSE in both wind speed and direction estimation.
Dynamic trials showed higher errors.
These discrepancies may partly result from spatial variations in the wind field rather than from algorithmic degradation.
Because the ultrasonic anemometer assumes a uniform wind field, the moving UAV can traverse regions with wind conditions different from those measured by the anemometer, leading to mismatches between the estimated and ground-truth values.
Across all five trials, the proposed method consistently yielded the smallest RMSE, demonstrating that our method maintains high precision and robustness even in dynamic, unstructured outdoor flight scenarios.

\section{CONCLUSION AND FUTURE WORK}
\label{sec:conclusion}

In this work, we proposed a high-precision, wide-range wind estimation method relying solely on UAV onboard sensors. 
The system combines a DOB for external force estimation with a combination of TPS and polynomial regression models to map aerodynamic forces to wind vectors, eliminating the need for dedicated wind-sensing hardware. 
A custom-designed wind barrel enhances aerodynamic sensitivity, as confirmed by CFD and wind tunnel analyses.

The proposed method was extensively validated through wind tunnel, indoor, and outdoor flight experiments. 
In controlled wind tunnel tests up to \SI{10}{m/s}, it demonstrated RMSEs as low as \SI{0.06}{m/s} for speed and \ang{3.6} for direction under laminar airflow. 
Outdoor experiments confirmed accurate horizontal wind estimation, achieving RMSEs of \SI{0.29}{m/s} and \ang{6.0} for speed and direction, respectively, with strong correlation to ground-truth trends. 
Indoor dynamic trials further verified robust performance under varying maneuvers, providing reliable horizontal wind estimates and uniquely enabling vertical wind estimation with RMSEs below \SI{0.17}{m/s}, which is not provided by prior baselines.
These results demonstrate the feasibility, robustness, and broad applicability of the proposed method. 

Future work will explore (i) direct aerodynamic force-to-wind mapping to reduce dependence on wind tunnel calibration and enable fully calibration-free deployment, and (ii) validation against true vertical wind sources (e.g., 3D wind facilities) to further strengthen vertical wind estimation capability.
Furthermore, integrating the proposed wind estimator with performance-guaranteed safety controllers~\cite{yu2025review} or cooperative fault-tolerant architectures~\cite{yu2024fault} could significantly enhance UAV robustness in extreme wind events—e.g., gust rejection during formation flight or safe landing under sensor degradation.

\bibliography{reference}
\bibliographystyle{IEEEtran}

\end{document}